\documentclass[11pt]{amsart}
\usepackage{url}
\usepackage{psfig}
\numberwithin{equation}{section}
\numberwithin{figure}{section}
\numberwithin{table}{section}

\usepackage{amsmath}
\usepackage{amssymb}
\usepackage{graphicx}

\newtheorem{lemma}{Lemma}

\newtheorem{proposition}[lemma]{Proposition}

\newcommand{\sign} {{\mathrm{sign} \, }}

\newcommand{\spot}{\marginpar{$\bullet$\kern30pt}}

\begin{document}
\bibliographystyle{plain}

\begin{center}
{\Large {\bf  Models for dependent extremes using stable
mixtures}}
\end{center}
\bigskip

\begin{center}
Running/short title: Models for dependent extremes
\end{center}
\bigskip

\begin{center}
{\large Anne-Laure\ Foug\`eres} \\
{\it \'Equipe Modal'X, Universit\'e Paris X - Nanterre}\\
 {\large John P. Nolan }\\
 { \it Department of Mathematics and Statistics,
American University} \\
{\large  Holger \ Rootz\'en} \\
 { \it Department
  of Mathematics, Chalmers University of Technology}
\end{center}
  \vskip 6mm


\vskip 2mm \noindent {\bf Abstract:} This paper unifies and extends
results on a class of multivariate Extreme Value (EV) models studied
by Hougaard, Crowder, and Tawn. In these models both unconditional
and conditional distributions are EV, and all lower-dimensional
marginals and maxima belong to the class. This leads to substantial
economies of understanding, analysis and prediction. One
interpretation of the models is as size mixtures of EV
distributions, where the mixing is by positive stable distributions.
A second interpretation is as exponential-stable location mixtures
(for Gumbel) or as power-stable scale mixtures (for non-Gumbel EV
distributions). A third interpretation is through a Peaks over
Thresholds model with a positive stable intensity. The mixing
variables are used as a modeling tool and for better understanding
and model checking. We study extreme value analogues of components
of variance models, and new time series, spatial, and continuous
parameter models for extreme values. The results are applied to data
from a pitting corrosion investigation.


\section{Introduction}
Multivariate models for extreme value data are attracting
substantial interest, see e.g. Kotz and Nadarajah (2000) and
Foug\`eres (2004). However, with the exception of Smith (2004) and
Heffernan and Tawn (2004), few applications involving more than two
or three dimensions have been reported. One main application area is
environmental extremes.  Dependence between extreme wind speeds and
rain fall can be important for reservoir safety (Anderson and
Nadarajah (1993), Ledford and Tawn (1996)), high mean water levels
occurring together with extreme waves may cause flooding (Bruun and
Tawn (1998), de Haan and de Ronde (1998)), and simultaneous high
water levels at different spatial locations pose risks for large
floods (Coles and Tawn (1991)). Another set of applications is in
economics where multivariate extreme value theory has been used to
model the risk that extreme fluctuations of several exchange rates
or of prices of several assets, such as stocks, occur together
(Mikosch (2004), Smith (2004), St\u{a}ric\u{a} (1999)).  A third
use, perhaps somewhat unlikely, is in the theory of rational choice
(McFadden (1978)). Below we will also consider a fourth problem,
analysis of pitting corrosion measurements (Kowaka (1994), Scarf and
Laycock (1994)).

The papers cited above all use  multivariate Extreme Value (EV)
distributions. The rationale is the ``extreme value argument'':
maxima of many individually small variables often have approximately
a (univariate or multivariate as the case may be) extreme value
distribution. However in  ``random effects'' situations this
argument becomes less clear. Suppose e.g. a number of groups each
has its own i.i.d variation but in addition each group is affected
by some overall random effect. Then, is it the unconditional
distributions which belong to the extreme value family, or is it the
conditional distribution, given the value of the random effect? In
many situations the extreme value argument seems equally compelling
for unconditional and conditional distributions. So, should one use
an EV model for the conditional distribution; or is it perhaps the
unconditional distributions which are extreme value?

In the present paper this problem is overcome by using models where
both conditional and unconditional distributions are EV. The models
have the further attractive properties that all lower-dimensional
marginals belong to the same class of models, and that maxima of all
kinds, e.g. over a number of ``groups'' with differing numbers of
elements, also have distributions which belong to the class.

The models are obtained by mixing EV distributions over a positive
stable distribution. They were first noted by Watson and Smith
(1985) and, in a survival analysis context, apparently independently
introduced by Hougaard (1986) and Crowder (1989). Further
interesting applications of such models were made in Crowder (1998).
The most general versions of these distributions were called the
asymmetric logistic distribution and the nested logistic
distribution by Tawn (1990) and McFadden (1978) and were further
studied in Coles and Tawn (1991). Crowder (1985) and Crowder and
Kimber (1997) contain some related material. However, we believe
that the full potential of these models is still far from being
realized. In this paper we have attempted to take three more steps
towards making them more widely useful.

The first step is to revisit the papers of Hougaard, Crowder and
Tawn, to collect and solidify the results in these papers. We
concentrated on two parts: the physical motivation for the models,
and a clear mathematical formulation of the general results.
The second step is to use the stable mixing variables not just as a
``trick'' to obtain multivariate distributions, but as a modeling
tool. Insights obtained from taking the mixing variable seriously
are new model checking tools, and better understanding of
identifiability of parameters and of the model in general.
The final important step is the realization that through suitable
choices of the mixing variables it is possible to obtain new
natural time series models, spatial models, and continuous
parameter models for extreme value data.  This  provides classes
of models for extreme value data which go beyond dimensions two
and three.

It is not immediately obvious from the forms of the asymmetric
logistic distribution and the nested logistic  distribution how to
simulate values from them, see e.g. Kotz and Nadarajah (1999,
Section 3.7). However the representation as stable mixtures makes
simulation straightforward. According to it, one can first simulate
the stable variables, using e.g. the method of Chambers {\it et al.}
(1976), and then simulate independent variables from the conditional
distribution given the stable variables, cf. Stephenson (2003). This
adds to the usefulness of the models.

Our results can be presented in two closely related ways: as mixture
models for Gumbel distributions, and as mixture models for the
general family of EV distributions. We first present the results for
Gumbel distributions. The Gumbel distribution has a special
importance. It occurs as the limit of maxima of most standard
distributions, specifically so for the normal distribution. In fact,
it is the only possible limit for the entire range of tail behavior
between polynomial decrease and (essentially) a finite endpoint.
Another reason is the approximate lack of memory property of the
locally exponential tails, which goes together with it. The Gumbel
distribution is known to fit well in many situations, e.g. for pit
corrosion measurements (Kowaka (1994)).

We present three motivations/interpretations of  the Gumbel models.
One is as an exponential-stable location mixture of independent
Gumbel distributions with the same scale parameter. A second
interpretation is as size mixtures of extreme value distributions,
where the mixing is by positive stable distributions. A third
interpretation is through a Peaks over Thresholds (PoT) model with a
positive stable random intensity.

We also develop the models in the general EV setting. In it, two out
of three physical motivations for the model, as size mixtures and as
maxima in a Peaks over Thresholds  model with a doubly stochastic
Poisson number of large values are the same as for the Gumbel model.
The counterpart to the remaining Gumbel interpretation, as a
location parameter mixture, is that the multivariate EV
distributions are obtained as scale mixtures with an accompanying
location change which keeps the endpoints of the distributions
fixed.

The basic motivations and explanations of the models for the Gumbel
case are collected in Section \ref{sect:mixtures} below. In Section
\ref{sect:newclasses} we rederive and remotivate the asymmetric and
nested logistic multivariate Gumbel distributions and introduce new
classes of multivariate Gumbel models for time series, spatial, and
continuous parameter applications. In Section \ref{sect:pitting} we
discuss estimation in the random effects model and in a hidden MA(1)
model. These models are then used to analyze a data set coming from
an investigation of
pitting corrosion on the lower hemflange of a car door. The section
also uses new model checking tools. Properties of the
exponential-stable mixing distributions are given in Section
\ref{sect:properties}. Section \ref{sect:evmixtures} translates the
Gumbel results and models from Sections \ref{sect:mixtures},
\ref{sect:newclasses}, and \ref{sect:pitting} to the general EV
family. Section \ref{sect:discussion} contains a small concluding
discussion.


\section{Mixtures of Gumbel distributions}\label{sect:mixtures}

In this section we revisit the physical
motivations/interpretations for the models, and add one of our
own - as a ``size mixture''. We present the motivations in a new
setting which seems particularly illustrative. This situation is
a standard type of pitting corrosion measurement. In it a number
of metal test specimens, e.g. from the body of a car, are divided
up into subareas, called test areas, and the deepest corrosion
pit in each of the test areas is measured. The presumption is
that there may be an extra variation between specimens (due to
position) which is not present between test areas from the same
specimen. In Section 4 we analyze such an experiment. One 
cause of extra variation in this experiment was the
randomness in the proportion of the surface which was covered by
corrosion-preventing 
coating. There undoubtedly were
other causes, such as differences in exposure to dirt and salt.
However, for the present purposes of illustration we mainly talk
about the variation in the size of the surface cover.

We introduce the ideas in the one-dimensional case. The motivations,
however, extend directly to the new multivariate models which are
treated in subsequent sections and which are the main interest of
this paper.

The mathematical basis  is the following
observation. Let $S$ be a standard positive $\alpha$-stable variable,
specified by its Laplace transform
\begin{equation}
\label{stable}
E(e^{-tS})=e^{-t^\alpha}, \;\;\;\; t \geq 0,
\end{equation}
 where necessarily $\alpha \in (0,1]$. (When $\alpha=1$, $S$ is taken to
 be identically 1, see the discussion in Section 5.) Further, let the
random variable $X$ be Gumbel distributed conditionally on $S$,
\begin{equation}
\label{xconditional}
P(X \leq x|S) = \exp(-Se^{-\frac{x-\mu}{\sigma}})
= \exp(-e^{-\frac{x-(\mu+\sigma \log(S))}{\sigma}}).
\end{equation}
Then by (\ref{stable}),
\begin{equation}
\label{xunconditional}
P(X \leq x)= \exp(-(e^{-\frac{x-\mu}{\sigma}})^\alpha)
= \exp(-e^{-\frac{x-\mu}{\sigma/\alpha}}).
\end{equation}
Hence unconditionally $X$ also has a Gumbel distribution, but the
mixing increases the scale parameter $\sigma$ of the conditional
Gumbel distribution to $\sigma/\alpha$.

We will sometimes use the terminology that the distribution of $X$ is
{\em directed} by the stable variable $S$. Let
$G \sim$ Gumbel$(\mu,\sigma)$ mean
that the random variable $G$ has the distribution function (d.f.)
$\exp(-e^{-\frac{x-\mu}{\sigma}})$. If $S$ has the distribution
specified by (\ref{stable}), the variable $M = \mu + \sigma \log(S)$
will be called exponential-stable with parameters $\alpha, \mu,$ and
$\sigma$. The symbols $M \sim$ ExpS$(\alpha,\mu, \sigma)$ will be used to
denote such a distribution.

We will give equation (\ref{xunconditional}) three different
interpretations. The first one was used by Crowder (1989) in the
context of a ``first order components of variance'' setting (cf also
Hougaard (1986)). The third one was put forth by Tawn (1990), and
discussed in a wind storm setting.

\vspace*{2mm}

(i) {\em Gumbel distribution as a location mixture of Gumbel
distributions:} If $G$ and $M$ are independent and $G \sim$
Gumbel$(\mu_1,\sigma)$ and $M \sim$ ExpS$(\alpha, \mu_2,\sigma)$
then $G+M \sim$ Gumbel$(\mu_1+\mu_2,\sigma/\alpha)$. This follows by
replacing $\mu$ in (\ref{xconditional}) and (\ref{xunconditional})
by $\mu_1+\mu_2$.


For the pitting corrosion measurements, the interpretation would be
that the maximal pit depth in a test area had a Gumbel distribution
with a random location parameter $\mu_1 + M$. The value of $M$ would
depend on the extent to which the test area  was exposed to
corrosion.

Briefly going beyond the one-dimensional model, it would be natural to
assume that different test areas would have different $G$-s but that
the variable M would be the same for all test areas on the same
specimen, and different for different test specimens.  A further remark
is that in this model it is not possible to separate $\mu_1$ and
$\mu_2$. However, the parameters can be made identifiable by assuming
that either $\mu_1$ or $\mu_2$ is zero.

\vspace*{2mm}

(ii) {\em Gumbel distribution as a size mixture of Gumbel
distributions:} If the maximum over a unit block has the Gumbel
d.f. $\exp(-e^{-\frac{x-\mu_1}{\sigma}})$ and blocks are independent
then the maximum over $n$ blocks, or equivalently over one block of
size $n$, has the d.f.
\begin{equation}\label{eq:maxn}
(\exp(-e^{-\frac{x-\mu_1}{\sigma}}))^n = \exp(-ne^{-\frac{x-\mu_1}{\sigma}}).
\end{equation}
In this equation it makes sense to think of non-integer block
sizes and random block sizes. In particular, it makes sense to
replace $n$ by $Se^{\mu_2/\sigma}$ in (\ref{eq:maxn}) to obtain
the d.f. $ \exp(-Se^{\mu_2/\sigma}e^{-\frac{x-\mu_1}{\sigma}})$.
It then again follows from (\ref{stable}) that the unconditional
distribution is Gumbel$(\mu_1+\mu_2,\sigma/\alpha)$. Thus the
Gumbel$(\mu_1 + \mu_2,\sigma/\alpha)$ distribution is obtained as
a ``size mixture'' of Gumbel$(\mu_1,\sigma)$ distributions, by
using the stable size distribution $S e^{\mu_2/\sigma}$. As
before, to make the model identifiable, one should assume that
either $\mu_1$ or $\mu_2$ is zero.

\vspace*{2mm}

The interpretation in the corrosion example is that $S
e^{\mu_2/\sigma}$ is the ``size'' of the part of the test area which
is exposed to corrosion. This size of course cannot be negative.
Further it could reasonably be expected to be determined as the sum
of many individually negligible contributions. Suitably interpreted,
these two properties together characterize the positive stable
distributions.

Next, it is well known that maxima of i.i.d. variables
asymptotically have a Gumbel distribution if the point process of
large values asymptotically is a Poisson process. More precisely, if
$\{Y_{n,i}\}$ are suitably linearly renormalized values of an i.i.d.
sequence $\{Y_i\}$ and $t_i=i/n$, then the point process $\sum_i
\epsilon_{(t_i, Y_{n,i})}$ tends to a Poisson process in the plane
with intensity $d\Lambda = dt \times d(e^{-(x-\mu)/\sigma})$ if and
only if the probability that $\max_{1\leq i \leq n} Y_{n,i} \leq x$
tends to $\exp(-e^{-\frac{x-\mu}{\sigma}})$, see e.g. Leadbetter et
al. (1983).
Our third
interpretation of the Gumbel mixture model is obtained by replacing
the constant intensity in the point process by a stable one.

\vspace*{2mm} {\it (iii) Gumbel distribution as the maximum of a
conditionally Poisson point process:} Suppose
$X$ is the maximum y-coordinate of a point process in $(0,1]\times R$
such that conditionally on a stable variable $S$ the point process
is Poisson with intensity $d\Lambda = S e^{\mu_2/\sigma} dt \times
d(e^{-(x-\mu_1)/\sigma})$. Then, by the same argument as above,
conditionally on $S$, the variable $X$ has d.f. $\exp(-S e^{\mu_2/\sigma}
e^{-\frac{x-\mu_1}{\sigma}})$, and as for (\ref{xunconditional}), it
follows that the unconditional distribution of $X$ is
Gumbel$(\mu_1+\mu_2,\sigma/\alpha)$.

\vspace*{2mm}

In the corrosion example,  the points in the point process
correspond to pit depths on the surface of the test area. The random
intensity $S e^{\mu_2/\sigma}$ then would describe an extra
stochastic variation in intensity of pits from test area to test
area. Again this has to be positive and perhaps is obtained as the
sum of many individually negligible influences, and hence approximately 
positive stable.  As above, one of $\mu_1$ or $\mu_2$ should be
assumed to be zero for identifiability.

It may also be noted that in some situations it may be possible to
use PoT observations, i.e. to actually observe the underlying  large
values, say all deep corrosion pits in each test area. Such
measurements could also be handled within the present framework, by
substituting the likelihoods in this paper with the corresponding
point process (or PoT) likelihoods. However, we will not pursue this
further here.

By way of comment, the logarithm of  the positive stable
distribution which occurs in the location mixture (i) has finite
moments of all orders. In contrast, the positive stable variables
themselves have infinite means. This, however, seems largely
irrelevant both for the mathematics of the models and for modelling.

\section{New classes of Gumbel processes}\label{sect:newclasses}
In this section we introduce a number of concrete Gumbel models
directed by linear stable processes: a random effects model, time
series models with directing stable linear processes, and a spatial
model with a stable moving average as directing process. We also
consider a hierarchical setup and continuous parameter models.

However, to provide a solid foundation for this paper and future
developments, we first give a precise mathematical formulation of
results of Tawn (1990). This shows the exact relations between the
three interpretations given in Section \ref{sect:mixtures} in a
general setting, and slightly generalizes (a restriction on the size
of the set $A$ is removed) Tawn's main result.

Let $T$ and $A$ be discrete index sets, where in addition $T$ is
assumed to be finite. Further let $\{c_{t,a}\}$ be non-negative
constants and let $\{S_a, a \in A \}$ be independent positive
$\alpha$-stable variables with distribution specified by
(\ref{stable}). We assume without further comment that $\sum_{a \in
A} c_{t,a} S_a$ converges almost surely for each $t$.

\begin{proposition}
Consider the following three models: \\
\noindent (i) $\displaystyle{X_t = G_t + \sigma_t \log(\sum_{a \in
A}c_{t,a}S_a)}, \; t\in T,$ where $G_t \sim$Gumbel$(\mu_t,
\sigma_t)$, and the $G_t$-s and $S_a$-s all are mutually
independent.

\vspace*{2mm}

\noindent (ii)  ${X_t, t \in T}$ are conditionally independent
random variables given ${S_a, a \in A}$, with marginal distributions
\begin{equation}
\label{generalconditional}
P(X_t \leq x_t | S_a, a \in A) = \exp\left(-(\sum_{a \in A}c_{t,a}S_a
) e^{-\frac{x_t-\mu_t}{\sigma_t}} \right), \;\;\;\; t \in T.
\end{equation}

\vspace*{2mm}

\noindent (iii)  For $t \in T$, $X_t$ is the maximum y-coordinate of
a point process in $(0,1]\times R$ such that conditionally on $S_a,
a \in A$ the point process is independent and Poisson with intensity
$\left(\sum_{a \in A}c_{t,a}S_a \right)dt \times
d(e^{-(x-\mu_t)/\sigma_t})$.

\vspace*{2mm}

Then all three models are the same, i.e. they have the same finite dimensional distributions:
\begin{equation}
\label{general} P(X_t \leq x_t, t \in T) = \prod_{a \in A}
\exp\left(-(\sum_{t \in T}c_{t,a} e^{-\frac{x_t-\mu_t}
{\sigma_t}})^\alpha\right),
\end{equation}
and this distribution is a multivariate extreme value distribution.
\end{proposition}

\begin{proof}By the form of the Gumbel distribution function, (i)
implies that (ii) holds. Similarly, by the same argument as for
(iii) of Section 2 above, it follows that (iii) of the proposition
implies (ii). Further, that (ii) implies (\ref{general}) follows
immediately from (\ref{stable}) since, by independence of the
$\{S_a\}$,
\begin{eqnarray*} P(X_t \leq
x_t, t \in T) &=& E\left(\exp(-\sum_{t \in T}\sum_{a \in
A}c_{t,a}S_a e^{-\frac{x_t-\mu_t}{\sigma_t}})\right)\\
 &=&\prod_{a \in
A} E\left(\exp[-S_a(\sum_{t \in T}c_{t,a}
e^{-\frac{x_t-\mu_t}{\sigma_t}})]\right).
\end{eqnarray*}
It is obvious that the distribution (\ref{general}) is max-stable,
and hence an EV distribution.
\end{proof}


As discussed in the introduction, a class of multivariate extreme
value mixture models is most useful if (a) both unconditional and
conditional distributions are extreme value, (b) lower-dimensional
marginal distributions also belong to the class, and (c) maxima over
any subsets have joint distributions which belong to the class. Now,
(a) is a part of Proposition 1. Further, if one sets some of the
$x_t$ in (\ref{general}) equal to infinity the corresponding terms
in the sum in the right hand side vanishes, but the expression still
is of the same general form, and hence the model (\ref{general})
satisfies the requirement (b).

The model also satisfies (c) if one imposes the extra restriction
that all the scale parameters have the same value, i.e. that
$\sigma_t = \sigma$, for $t \in T$. For the marginal distribution
of a maximum this is because if $T_1 \subset T$ then
\begin{equation}
\label{maxstable} P(\max_{t \in T_1} X_t \leq x) = \prod_{a \in A}
\exp\left(-(\sum_{t \in T_1} c_{t,a}
e^{\frac{\mu_t}{\sigma}})^\alpha
e^{-\frac{x}{\sigma/\alpha}}\right),
\end{equation}
or equivalently
$$ \max_{t \in T_1} X_t \sim \mathrm{Gumbel}\left( (\sigma/\alpha)
\log(\sum_{a \in A}(\sum_{t \in
T_1}c_{t,a}e^{\mu_t/\sigma})^\alpha), \sigma/\alpha\right).
$$
In particular, by letting $T_1$ be a one point set we see that in
this case marginals are Gumbel distributed,
$$ X_t \sim \mathrm{Gumbel}\left( (\sigma/\alpha) \log(\sum_{a \in
 A}(c_{t,a}e^{\mu_t/\sigma})^\alpha), \sigma/\alpha\right).
$$

Moreover, joint distributions of
maxima also  belong to the class (\ref{general}) of distributions.
E.g. let $T_1$ and $T_2$ be disjoint subsets of $T$ and set
$c_{T_i,a} = \sum_{t \in T_i}c_{t,a} \exp{(\mu_t/\sigma}) $, for
$i=1,2$.
Then, as can be seen from (\ref{generalconditional}) or
(\ref{general}),
\begin{equation*}
 P(\max_{t \in T_1} X_t \leq x_1, \max_{t \in T_2} X_t \leq x_2) = \prod_{a \in A}
\exp\left(-(c_{T_1,a} e^{-\frac{x_1}{\sigma}}+c_{T_2,a}
e^{-\frac{x_2}{\sigma}})^\alpha \right),
\end{equation*}
which has the form (\ref{general}). Similar but more complicated
formulas hold when more subsets are involved and when the subsets
can overlap.

These two properties are touched upon by Crowder (1989) in a
less general situation, and also by Tawn (1990).

Conditions (i) - (iii) in Proposition 1  correspond to the three
``physical'' interpretations in Section 2. We now turn to a number
of specific models. Which interpretation is most relevant of course
varies from model to model. E.g the first model below is the
standard logistic model for extreme value data, but with the
interpretation as a random effects model. We will use it on a pit
corrosion example, where perhaps the interpretation (ii) is most
compelling. However, to streamline presentation, we will for the
rest of this section formulate the models as in (i), but of course
could equally well have used (ii) or (iii).

\vspace*{2mm} \noindent {\em Example: A one-way random effects
model.} This is the model
\begin{equation}
\label{eq.oneway.cdf}
X_{i,j} = \mu + \tau_i + G_{i,j}, \;\;\;\;\;\;\ 1 \leq i \leq m, \; 1 \leq j \leq n_i
\end{equation}
with $\mu$ a constant, $\tau_i \sim$ExpS$(\alpha, 0, \sigma)$, $G_{i,j}
\sim$Gumbel$(0,\sigma)$ and all variables independent.

Setting  $T=\{(i,j); \ 1 \leq i \leq m, \; 1 \leq j \leq n_i \},$  $
A = \{1, 2, \dots, m \}$ and $c_{(i,j),k} = 1_{\{i=k\}}$, this is a
special case of the situation in Proposition 1 and we directly get
the distribution function
\begin{equation}\label{one.way.loglik}
P(X_{i,j} \leq x_{i,j}, \ 1 \leq i \leq m, \; 1 \leq j \leq n_i)
= \prod_{i=1}^m \exp(- (\sum_{j=1}^{n_i} e^{-\frac{x_{i,j}-\mu}{\sigma}})^\alpha).
\end{equation}

According to Proposition 1 and the subsequent remarks this is a
multivariate EV distribution, and explicit formulas are directly
available for the distribution of all kinds of unconditional and
conditional maxima. In particular the marginal distributions are
Gumbel$(\mu, \sigma^*)$ for $\sigma^* = \sigma/\alpha$. \hfill
{$\Box$}

\hspace*{2mm}

This model can be extended to higher order random effects models
which are ``linear on an exponential scale''. It can also be
natural, for instance in a ``repeated measurements'' setting, to
let $\mu$ be a function of $t$, perhaps depending on the values of
known covariates, as done in Crowder (1989, 1998) or Hougaard
(1986). Note however that in the context of repeated data, say
$(Y_1, \cdots, Y_p)$, the set $T$ from Proposition 1 has to be
$T=\{(i,j); 1\leq i \leq p, 1\leq j \leq n_i\}$, whereas we allow
more general $T$'s.

We next turn to time series models.  A linear stationary positive
stable process may be obtained as $H_t=\sum_{i=-\infty}^\infty b_i
S_{t-i}$, where the $S_i$ have distribution (\ref{stable}), the
$b_i$ are nonnegative constants, and the sum converges in distribution if $\sum
b_i^\alpha < \infty$. Defining
\begin{equation}
\label{linear}
X_t = \mu_t + \sigma \log
(H_t) + G_t,
\end{equation}
for some constants $\mu_t$ gives a Gumbel time series model. In
particular (\ref{linear}) includes hidden ARMA models. We next
look closer at the two simplest cases of this.

\vspace*{2mm}

 {\em Example: A hidden MA-process model.} Suppose $H_t= b_0 S_t+b_1
S_{t-1}+ \dots +b_q S_{t-q}$ and $X_t$ is defined by (\ref{linear}),
where the $S_i$ have distribution (\ref{stable}), $G_t
\sim$Gumbel$(0, \sigma)$ and all variables are mutually independent.
Then, by Proposition 1 with $T=\{1, \dots, n\}$, and $A=\{0, \pm 1,
\dots \}$,

\begin{equation}
\label{eq:ma}
P(X_t \leq x_t, \; 1 \leq t \leq n) = \prod_{k=1-q}^n
\exp(-(\sum_{t=1
 \vee k}^{n \wedge (k+q)}
b_{t-k}e^{-\frac{x_t-\mu_t}{\sigma}})^\alpha).
\end{equation}
\hfill $\Box$

\vspace*{2mm}

 {\em Example: A hidden AR-process model.} For $0<\rho<1$ define the
positive stable AR-process $H_t$ by $H_t=\sum_{i=0}^\infty \rho^i
S_{t-i}$, and let $X_t$ be given by (\ref{linear}), with  the $S_i$
and $G_t$ as before.
From the definition of  $H_t$,
\begin{eqnarray}
\label{AR}
H_0 &=& \sum_{i=0}^\infty \rho^i S_{-i} \\
H_1 &=& \rho H_0 + S_1 \nonumber\\
&\vdots& \nonumber \\
H_n &=& \rho^n H_0 + \rho^{n-1}S_1 + \cdots +\rho S_{n-1} + S_n,
\nonumber
\end{eqnarray}
and in addition, by (\ref{stable}) $H_0$ has the same distribution as
$$
(\sum_{i=0}^\infty \rho^{i\alpha})^{1/\alpha} S_0 = (1-\rho^\alpha)^{-1/\alpha}S_0,
$$ and is independent of $S_1, \dots, S_n$. Thus, the model is again of
the form considered in Proposition 1, with $T=\{0, \dots, n\}$,
$A=\{0, \pm 1, \dots \}$ and $c_{t,0}= \rho^t
(1-\rho^\alpha)^{-1/\alpha}, \; c_{t,a} = \rho^{t-a}$ for $ a=1,
\dots, t$ and $c_{t,a} = 0$ otherwise. Thus by Proposition 1 the
distribution function is
$$
P(X_t \leq x_t, \; 0 \leq t \leq n)
= \exp[-(1-\rho^\alpha)^{-1}(\sum_{t=0}^n \rho^t e^{-\frac{x_t-\mu_t}{\sigma}})^\alpha]
 \prod_{i=1}^n \exp(-(\sum_{t=i}^n \rho^{t-i}e^{-\frac{x_t-\mu_t}{\sigma}})^\alpha).
$$
\hfill $\Box$

\hspace*{2mm}

In the next example we consider models on the integer lattice in the
plane. Let $n_{(i,j)}$ be a system of neighborhoods with the standard
properties $(i,j) \in n_{(i,j)}$ and $(k,l) \in n_{(i,j)}
\Leftrightarrow (i,j) \in n_{(k,l)}$. A simple example is when the
neighbors are the four closest points and the point itself, i.e. when
$n_{(i,j)} = \{(i,j),(i-1,j),(i+1,j),(i,j-1),(i,j+1)\}$.

\vspace*{2mm} {\em Example: A spatial hidden MA-process model.} Let
$\{S_{i,j}; - \infty < i,j < \infty \}$ be independent standard
positive $\alpha$-stable variables and set $H_{i,j}=\sum_{(k,l)
\in n_{(i,j)}} \delta S_{k,l}$ where $\delta$ is a positive constant.
Put
$$
X_{i,j} = \mu_{i,j} +\sigma \log(H_{i,j}) + G_{i,j}, \;\;\;\; 1 \leq i,j \leq n,
$$ where the $G_{i,j}$ are mutually independent and independent of the
$S_{i,j}$, and $G_{i,j} \sim$ Gumbel$(0,\sigma)$. Again this is of the
form considered in Proposition 1, now with $c_{(i,j),(k,l)} = \delta$ if $(i,j)
\in n_{(k,l)}$ and zero otherwise.  To write down the joint
distribution function it is convenient to use the notation
$\bar{n}_{(k,l)} = {n}_{(k,l)} \cap \{(i,j); \;\; 1 \leq i,j \leq
n \}$. We then get that

$$
P(X_{i,j} \leq x_{i,j}; \; 1 \leq i,j \leq n)
= \prod_{(k,l)} \exp(-\delta^\alpha (\sum_{(i,j) \in \bar{n}_{(k,l)}}
e^{-\frac{x_{i,j} - \mu_{i,j}}{\sigma}})^\alpha).$$
\hfill $\Box$

\hspace*{2mm}

We now turn to a situation not covered by Proposition 1, the
so-called nested logistic model of McFadden (see
Tawn (1990)).

{\em Example: A two-layer hierarchical model.}
Consider  the model
$$ X_{i,j,k} = \mu + \tau_i + \eta_{i,j} + G_{i,j,k}, \;\;\;\;\;\;\ 1
\leq i \leq m, \; 1 \leq j \leq n_i, \; 1 \leq k \leq r_{i,j},
$$ with $\mu$ a constant, $\tau_i \sim $ExpS$(\beta, 0, \sigma/\alpha)^{1/\alpha}$,
$\eta_{i,j} \sim $ExpS$(\beta, 0,\sigma)$, $G_{i,j,k}
\sim$Gumbel$(0,\sigma)$, and all variables independent. By repeated
conditioning we obtain, after some calculations similar to the proof
of Proposition 1,
$$ P(X_{i,j,k} \leq x_{i,j,k}, \ 1 \leq i \leq m, \; 1 \leq j \leq n_i,
\; 1 \leq k \leq r_{i,j}) $$
$$
= \prod_{i=1}^m \exp[- \{\sum_{j=1}^{n_i}
(\sum_{k=1}^{r_{i,j}} e^{-\frac{x_{i,j,k}-\mu}{\sigma}})^\alpha\}^\beta].
$$
\hfill $\Box$

\hspace*{2mm}

There also are continuous parameter versions of Proposition 1. Let
$\{S_j(\mathbf{s}); \mathbf{s} \in R^k\}$ be independently scattered
positive stable noise (see Samorodnitsky and Taqqu (1994, Chapter 3)).  We
assume that the noise is standardized, so that for nonnegative functions $f \in
L_\alpha$,
\begin{equation}
\label{stablecont}
E[\exp\{-\int_{-\infty}^\infty f(\mathbf{s}) S_j(d\mathbf{s})\}] =
\exp(-\int_{-\infty}^\infty f(\mathbf{s})^\alpha d\mathbf{s}).
\end{equation}
In the sequel we will without comment assume that functions
$f(\cdot)$ are such that integrals converge, and integrals are
taken to be over $R^k$.

\begin{proposition}
Suppose that there are nonnegative functions $f_j(\mathbf{t},\mathbf{s})$ with
$\mathbf{t} \in R^\ell$,  $ \mathbf{s} \in R^k$ such that
$$ X_\mathbf{t} = G_\mathbf{t} + \sigma_\mathbf{t} \log(\sum_{j=1}^m \int
f_j(\mathbf{t},\mathbf{s}) S_j(d\mathbf{s})), \;\;\;\; \mathbf{t} =
\mathbf{t}_1, \dots, \mathbf{t}_n,
$$ where $G_t \sim$Gumbel$(\mu, \sigma_t)$, and all variables are
mutually independent. Then
\begin{equation}
\label{continuous} P(X_{\mathbf{t}_i} \leq x_{\mathbf{t}_i}; \; i =
1, \dots, n) =\prod_{j=1}^m \exp(- \int (\sum_{i=1}^n
f_j(\mathbf{t}_i,\mathbf{s})e^{-\frac{x_{\mathbf{t}_i} -
\mu_{\mathbf{t}_i}}{\sigma_{\mathbf{t}_i}}})^\alpha d\mathbf{s}).
\end{equation}
\end{proposition}

The proof follows from (\ref{stablecont}) in the same way as
Proposition 1 follows from (\ref{stable}). The interpretations (ii),
as size mixtures, and (iii) as a random Poisson intensity could
equally well have been used as assumptions. However, this we leave to
the reader.

Proposition 2 gives a natural model for environmental extremes, such
as yearly maximum wind speeds or water levels, at irregularly located
measuring stations. E.g. one could assume years to be independent and
obtain a simple isotropic model for one year by choosing $k=\ell=2,
\;\; m=1$ and $f_1(\mathbf{t},\mathbf{s})
=\exp(-d|\mathbf{t}-\mathbf{s}|^\beta)$, for some constants $d, \beta
>0$. One extension to non-isotropic situations is by letting $D$ be a
diagonal matrix with positive diagonal elements and taking
$f_1(\mathbf{t},\mathbf{s})
=\exp(-(\mathbf{t}-\mathbf{s})^tD(\mathbf{t}-\mathbf{s})\beta)$.
(Formally the entire distribution function for $n$ years is also of
the form (\ref{continuous}), as can be seen by taking $\ell=3, \;
m=n$ and letting the different $S_j$ correspond to different years.)
It is possible to derive recursion formulas for the densities of
these models in a similar but more complicated way as for the random
effects model. If the number of measuring stations is not too large,
these expressions may be numerically tractable. However, we will not
investigate this further in this paper.

\section{Data analysis}\label{sect:pitting}

In this section we illustrate the random effects model and the
hidden MA(1) model from Section \ref{sect:newclasses} by using them
to analyze a set of pit corrosion measurements. As preliminaries we
first discuss maximum likelihood estimation in the two models.

\subsection{Estimation in the random effects model}\label{subsect:onewayestimation}

Let $0 < \sigma < \sigma^*$, $-\infty < \mu^* < \infty$, so
$\alpha := \sigma/\sigma^* \in (0,1)$.
Assume a data set $\mathbf{X}$ that comes from $m$ groups,
\begin{eqnarray}
\textrm{group 1}: & X_{1,1},X_{1,2},\ldots,X_{1,n_1} \nonumber \\
\textrm{group 2}: & X_{2,1},X_{2,2},\ldots,X_{2,n_2} \label{eq:data.set} \\
                  & \hspace{1cm} \vdots  \nonumber \\
\textrm{group m}: & X_{m,1},X_{m,2},\ldots,X_{m,n_m}.  \nonumber
\end{eqnarray}
The groups are assumed to be independent and the $i^\mathrm{th}$ group
comes from a Gumbel$(0,\sigma)$ distribution, where the location
parameter $\mu_i$ for group $i$ is drawn from an ExpS$(\alpha=\sigma/\sigma^*,
\mu^*,\sigma)$ distribution.  The goal is to estimate the three
parameters $\mathbf{\theta}=(\sigma,\sigma^*,\mu^*)$ from the data by
maximum likelihood.

The likelihood $L(\mathbf{\theta}|\mathbf{X})= \prod_{i=1}^m
L_i(\mathbf{\theta}|X_{i,1},\ldots,X_{i,n_i})$ is the product of
the group likelihoods.  Each of these terms can be derived by
differentiating (\ref{one.way.loglik}) with respect to
$x_1,\ldots,x_n$. The direct calculations are complicated, but
Property (1) of Shi (1995) gives recursions for the likelihood
function for the group in terms of certain coefficients $\{
q_{n,j} \}$.

The maximum likelihood algorithm has been implemented in S-Plus/R.
The estimation procedure numerically evaluates
$\ell(\mathbf{\theta}|\mathbf{X}) = \log
L(\mathbf{\theta}|\mathbf{X})$ and numerically maximizes it to
find the estimate of $\mathbf{\theta}$.  The search is initialized
at $\mathbf{\theta}_0 := (\sigma_0/2, \sigma_0, \mu_0)$, where
$\mu_0$ and $\sigma_0$ are estimates of the Gumbel parameters for
the (ungrouped) data set $\mathbf{X}$.  This estimate is found by
using the probability-weighted moment estimator, see e.g.
Section~1.7.6 of Kotz and Nadarajah (2000).

Usually this makes it straightforward to find maximum likelihood
estimates by numerical optimization. However, if a group is large
or $\alpha$ is small, the coefficients in the recursion can be
very large. E.g. the constant term in Shi's notation is
$$q_{n,0} = \left( \frac {n-1}{\alpha} -1 \right) \left( \frac
{n-2}{\alpha} -1 \right) \left( \frac {n-3}{\alpha} -1 \right)
\cdots \left( \frac {1}{\alpha} -1 \right).$$ In some cases this can
cause numerical overflow in the optimization routines. Further, if
all groups only have one value or if there is only one group then
parameters are not identifiable. Presumably parameter estimates will
be unreliable if data is close to these situations. This however was
not the case for the corrosion data in Section 4.3 below. Besides,
we made rather many simulations (not included in the paper) from
both random effects Gumbel model and independent Gumbel model with
arbitrary means, and checked on these simulations that the maximum
likelihood estimators perform reasonably well, as soon as there are
a few groups, and even when some of the groups are rather small.

In passing we note an alternative way to derive the likelihood,
which in addition indicates a possibility to compute it by
simulation. A group likelihood, conditional on $\tau$, is
$$ \prod_{j=1}
^{n} \, \frac{1}{\sigma} e^{-\frac{x_{j}-\mu
-\tau}{\sigma}} \exp \left\{ - e^{-\frac{x_{j}-\mu
-\tau}{\sigma}} \right\} = \frac{1}{\sigma^{n}} S^{n} e^{-\sum_{j=1}^{n}
\frac{x_{j}-\mu}{\sigma}} \exp \left\{ - S
\sum_{j=1}^{n} e^{-\frac{x_{j}-\mu}{\sigma}}\right\}, $$
where $\tau=\sigma \log S$ and $S$ is a standard $\alpha-$stable
variable, as previously. Hence, a  group likelihood is
$$ \frac{1}{\sigma^{n}} e^{-\sum_{j=1}^{n} \frac{x_{j}-\mu}{\sigma}} E\left[
S^{n} \exp \left\{ - S \sum_{j=1}^{n}
e^{-\frac{x_{j}-\mu}{\sigma}}\right\} \right]. $$
Let $\Delta = \sum_{j=1}^{n}e^{-(x_{j}-\mu)/\sigma}$.
Then, the
expectation in the last expression reduces to
$$ E\left[ S^{n} e^{ - S \Delta}\right]
=
E\left[ \frac{d^{n}}{d\Delta^{n}} \left\{e^{ - S \Delta}
\right\}\right]
= (-1)^{n}
\frac{d^{n}}{d\Delta^{n}}
\left\{e^{-\Delta^{\alpha}}\right\},
$$ where the second equality makes one more use of the stable
distribution of $S$.

\subsection{Estimation in the hidden MA(1) model}\label{subsect:ma(1)estimation}

By (\ref{eq:ma}) the hidden MA(1) model with constant location
parameter, $\mu_t=\mu$ and, for identifiability, $b_0=1, b_1=b$ has
distribution function
\begin{equation}
\label{eq:ma(1).df}
F=P(X_t \leq x_t, \; 1 \leq t \leq n) =
\exp\left(-\left\{(bz_1)^\alpha +\sum_{t=1}^{n-1}
(z_t +bz_{t+1})^\alpha + z_n^\alpha\right\}\right),
\end{equation}
where $z_t = \exp(-(x_t - \mu)/\sigma)$. The parameters of the
model are $\mathbf{\theta} =(\mu, b, \sigma,\alpha)$. By differentiation
with respect to $x_1, \dots,x_n$ the likelihood
function can be seen to be of the form
$$ L(\mathbf{\theta}|\mathbf{X}) = Q_{n}F \prod_{t=1}^n {z_t \over \sigma},$$
with F from (\ref{eq:ma(1).df}) and $Q_n$ defined recursively as follows.
Set $u_1 = b z_1$, $u_t = z_{t-1} + b z_t$ for $t=2,\ldots,n$, $u_{n+1}=z_n$.
Then $F = \exp(-\sum_{t=1}^{n+1} u_t^\alpha)$ and
\begin{eqnarray*}
Q_0&=&1, \hspace{1cm} Q_1= \alpha \left(b u_1^{\alpha-1}+u_2^{\alpha-1}\right), \\
Q_i&=& -Q_{i-2}\alpha(\alpha-1)b u_{i}^{\alpha-2}
      +Q_{i-1}\alpha \left(b u_{i}^{\alpha-1}+u_{i+1}^{\alpha-1}\right), \quad i=2, \dots, n.
\end{eqnarray*}
When $b=0$, the $Q_1$ term above should be interpreted as $Q_1= \alpha u_2^{\alpha-1}$,
which makes the likelihood formula valid in the case where the $x_t$ are independent.

Maximum likelihood estimation of the parameters $(\mu,b,\sigma,\alpha)$
has been implemented in S-Plus/R, where
$$\log \{L(\theta|\mathbf{X})\}= \log Q_n -\sum_{t=1}^{n+1} u_t^\alpha -
\sum_{t=1}^n \left( \frac {x_t-\mu} {\sigma} \right) - n \log
\sigma$$ is computed and numerically maximized.  As default the
search is started at $(\mu=\mu_0,b=0, \sigma=\sigma_0/0.5,
\alpha=0.5)$, where $(\mu_0,\sigma_0)$ are the Gumbel
probability-weighted moment estimators for the data set. In ad hoc
simulations to test this method, we sometimes observed that results
were sensitive to the choice of starting values when the sample size
was small. Apparently the likelihood surface has local maxima in
such cases. To deal with this problem, we started the search at
several different randomly chosen points and chose as estimator the
final values which gave the highest likelihood.

\subsection{Pitting corrosion data analysis}\label{subsect:dataanalysis}

The pitting corrosion investigation which generated this data set
was briefly mentioned in the beginning of Section
\ref{sect:mixtures}. Specifically, pieces (or ``test specimens'')
were cut out from different parts of the bottom hemflange of the
aluminum back door of a twelve year old station wagon. The corrosion
products were dissolved from the pieces, and the deepest corrosion
pit was measured in a number of one centimetre long test areas on
each specimen. The hemflange had been glued together and had also
been treated with a corrosion preventing coating. Surface areas
where the glue or coating was intact showed no corrosion. However,
in some places the glue and coating had not penetrated well or had
fallen of, leaving the surface exposed to corrosion. The proportion
of the area which could corrode varied between specimens, and this
was a potential cause of extra variation in the corrosion
measurements. These areas, however, had not been measured (and it
would have been difficult to do so) and there were  other
causes of extra variation, such as varying exposure to salt.

Interest was centered on the risk of penetration by the deepest
corrosion pit on the outer surface of the hemflange. The data set
for this surface consisted of microscope measurements (in microns)
of the maximum pit depth in 11 to 15 test areas on each of 12
specimens. There was no corrosion on 5 of the test specimens, and on
one specimen only two test areas showed any corrosion. These 6
specimens were excluded from our analysis. Also in the remaining
specimens there were some corrosion free test areas, and the data we
used for analysis hence consisted of 6 groups (=test specimens) with
varying numbers (ranging from 4 to 14) of measured maximum pit
depths.

The engineers who performed the experiment disregarded the group
structure and considered the pooled data set as an i.i.d Gumbel
sample. The maximum likelihood parameter estimates under this model
were $(\mu_{\mathrm{pool}}, \sigma_{\mathrm{pool}}) = (145.6, 69.4)$. It
was remarked by the engineers that there seemed to be some deviation
from a straight line in Gumbel plot, see Figure \ref{fig:poolgumbel}.
While the overall fit to the pooled data seems reasonable, there is clear group structure.

\begin{center}
Include Figure 7.1 here
\end{center}

 We instead modeled and analyzed the data as
 dependent 4 to 14-dimensional random vectors. We first made use of
the random effects Gumbel model from Subsection
\ref{subsect:onewayestimation}. The aim was both to see if this
model fitted better and to check wether it lead to a substantially
different risk estimate. In addition to the extra variation between
test specimens there might also be a short range dependence between
neighboring test areas. We tried to judge the size of short range
dependence by fitting a hidden MA(1) model on top of the random
effects model.

The maximum likelihood estimates in the random effects Gumbel model
were $(\mu, \sigma, \alpha) = (140.9, 54.1,0.716)$ with standard
deviations $(21.75,5.71,0.118)$ estimated from the inverse of the
empirical information matrix. A very rough calculation of the risk
of perforation can then be made as follows. There are about 15 test
specimens on a hemflange. Let us assume, as was the case with the
present data, that typically about 6 of the test specimens will show
corrosion and that on average about 11 test areas on each specimen
will be corroded. Then, by (\ref{one.way.loglik}) the estimated
distribution function of the maximum pit depth for one car would be
$$\hat{F}(x) = \exp(-6(11e^{-{ {x-140.9}\over 54.1}})^{0.716}).$$
The thickness of the aluminum was 1.1 mm = 1100 microns and hence we
estimate that there on the average will be perforation in one out of
$1/(1-\hat{F}(1100)) = 9671$ cars. A delta method 95\% confidence
interval for this estimate is $(8392,10950)$. If we instead,
following the engineering analysis, use the pooled Gumbel model with
the assumption that typically there are $6 \times 11=66$ corroded
test areas on a hemflange, the risk estimate is that on the average
there is penetration in one out of
$(1-\exp(-66e^{-\frac{1100-145.6}{69.4}})^{-1} \approx \allowbreak
14374$ cars. A delta method 95\% confidence interval is
$(13115,15632)$. Thus, the random effects model gave a practically
and statistically significantly different answer than the pooled
analysis.

The formulation as a random effects model gives a number of
possibilities for model checking.  From Figure \ref{fig:groupgumbel}
can be seen that the Gumbel distribution fits reasonably well to the
separate groups, that there indeed seems to be an extra variation
between groups, and that the fitted lines are approximately
parallel. As a  formal check on this, we made a conditional
analysis, fitting separate Gumbel distributions to the groups by
maximum likelihood. In this we considered three different models,
the first with separate $\mu$-s and $\sigma$-s for the groups, the
second with all groups assumed to have the same $\sigma$ but
different $\mu$'s for the different groups, and a third model with
the same $\sigma$ and the same $\mu$ for all observations. A
likelihood ratio test between the first two models gave $p=.53$, and
hence it seemed reasonable to assume the same $\sigma$ in all
groups, as is done in the random effects model. A test of the second
model against third lead to $p=2\cdot 10^{-6}$. Thus the pooled
model is rejected, while this analysis did not contradict the
validity of the random effects model.

As further checks on the random effects model, the $\sigma$ estimate
from the second model in the previous paragraph was $47.6$, which is
reasonably close to the $\sigma$ estimate $54.1$ in the random
effects model. Similarly, $\sigma^* = 75.6$ and
$\sigma_{\mathrm{pool}}=69.4$ are rather close, as they should be. A
further comparison is that the correlation coefficient estimated
nonparametrically from the data was $0.44$. This can be compared
with the correlation coefficient $1-\hat{\alpha}^2 = 0.49$ computed
from the fitted model.
\begin{center}
Include Figure 7.2 here
\end{center}
 Figure \ref{fig:qqmu.stable} shows the quantiles of the estimated $\mu$-s
against the quantiles of the fitted exponential-stable distribution.
According to the model, the $\mu$-s are exponential-stable, and
hence, apart from estimation error, the estimated $\mu$-s are
expected to be exponential-stable, so this plot is a diagnostic for
the fit of the mixing distribution. The plot  also shows a
reasonable fit, and in fact looks much like the same qq-plots from
simulated values from the model. Thus, neither of these model checks
indicated problems with the random effects model.

\begin{center}
Include Figure 7.3 here
\end{center}

As a final analysis we fitted the hidden MA(1) model from Subsection
\ref{subsect:ma(1)estimation} to the data, since there was a
possibility of extra dependence between neighboring test areas. In
this we assumed groups were independent and had their own $\mu$-s,
but that $\sigma, \alpha$ and $b$ were the same in all 6 groups.
Thus there were in all 9 parameters, the six group means
$\mu_1,\mu_2,\mu_3,\mu_4,\mu_5,\mu_6$ and the parameters
$\sigma,\alpha,b$. Maximum likelihood estimation using the default
initial values got stuck in a local maximum, and we hence did the
optimization for 100 different starting values for
$\sigma,\alpha,b$, chosen at random from the cube $[7,
54]\times[0.1, 0.99]\times[0, 2]$. As estimates we took the final
values which gave the highest likelihood. For the $\mu$-s in the 6
groups these were $87.3, 142.0, 132.4, 140.0,67.6,214.8$ and the
estimators for the remaining parameters were
$\hat{\sigma}=29.6,\hat{\alpha}=0.58,\hat{b}=0.13$.

From the model, the marginal distributions in the groups are Gumbel
with location parameter $\mu+\frac{\sigma}{\alpha}\log(1+b^\alpha)$
and scale parameter $\sigma/\alpha$. The estimates of these agreed
to within 5\% with their initial values, which indicated that these
parameters were reasonably well determined by the data. The
remaining two parameters, $\alpha$ and $b$, model the dependence
structure. The smaller the $\alpha$ and the closer $b$ is to one,
the higher is the dependence. These parameters seemed harder to
estimate. However, their estimated values indicated a rather weak
local dependence, and did not contradict the validity of the random
effects model.

We accordingly stopped the analysis at this point. If the dependence
had been judged important, we could have tried to fit a model which
included both random group means and a local MA(1) dependence.
Further model checking, as suggested by Crowder (1989, Section 3.3),
could be performed by using the probability integral transform
marginally to get uniform (but dependent) residuals or by computing
Rosenblatt residuals which are approximately independent if the
model is correct.

In summary: The pooled analysis did not fit the data and lead to
significantly different results than the random effects model.
Instead the random effects model seemed to give a good
representation of the data -- in particular none of the several
diagnostic checks indicated serious departures from it -- and we
believe it led to credible estimates. By way of further comment, it
may be noted that we obtained a successful fit of the hidden MA(1)
model, and that it produced useful information.

A weak point in the analysis is the assumption that a hemflange has
$6$ test specimens with $11$ corroded test areas each. Further the
variation in pit depths from car to car is not included in the data.
If measurements on several cars had been available, it would have
been natural to try to fit the hierarchical model from Section
\ref{sect:newclasses}.


\section{Some properties of the mixing distribution}\label{sect:properties}

This section discusses some of the basic facts about the models. In
the notation of Samorodnitsky and Taqqu (1994), the r.v. $S$ in
(\ref{stable}) is $S_{\alpha}( (\cos \pi \alpha/2)^{1/\alpha},1,0)$;
in the notation of Zolotarev (1986), $S \sim S_C(\alpha,1,1)$. It
has characteristic function
$$ E \exp(i t S) =
  \exp \left\{-\cos(\pi \alpha/2) |t|^{\alpha}
         \left[1-i \tan(\pi \alpha /2) (\sign t)\right] \right\}.  $$
Let $F_S(s)$ be the d.f. and $f_S(s)$ be the density of $S$. If $M
\sim $ExpS$(\alpha,\mu,\sigma)$, then the d.f. and density of $M$
are $F_M(x) = F_S[ \exp\{(x-\mu)/\sigma\} ]$ and $f_M(x) =
\exp\{(x-\mu)/\sigma\} f_S[ \exp\{(x-\mu)/\sigma\}]/\sigma$. Using
the programs for computing with stable distributions described in
Nolan (1997), it is possible to compute densities, d.f., quantiles
and simulate values for $M$. Figure~\ref{fig:expstable.pdf} shows
the density of some log-stable distributions. The densities all have
support $(-\infty,\infty)$ and appear to be unimodal. Note that as
$\alpha \uparrow 1$, $S$ converges in distribution to 1 and hence
$M=\log S$ converges in distribution to 0.

\begin{center}
{\it Include Figure 7.4 here}
\end{center}
It is well-known that the upper tail of $S$ is asymptotically
Pareto: as $x \rightarrow \infty$, $P(S > x) \sim c_\alpha
x^{-\alpha}$ where $c_\alpha=\Gamma(\alpha) \sin(\pi \alpha)/\pi$.
This implies that the right tail of $M
\sim$ExpS$(\alpha,\mu,\sigma)$ is asymptotically exponential: as $t
\rightarrow \infty$,
$$P(M > t) = P\left(S > \exp\left({t-\mu \over \sigma} \right)\right)
       \sim c_\alpha \exp\left(-{t-\mu \over \sigma/\alpha} \right).$$
The left tail of $S$ is light, see e.g. Section 2.5 of Zolotarev
(1986), so the left tail of $M$ is even lighter. Thus all moments of
$M$ exist; in particular, using the results of Section 3.6 of
Zolotarev (1986),
$$\mathrm{E}(M) = \mu + \sigma \gamma_{Euler} \left( {1 \over \alpha} - 1 \right), \quad
\mathrm{Var}(M) = {\pi^2 \sigma^2 \over 6} \left( {1 \over \alpha^2}
-1 \right),$$ where $\gamma_{Euler} \approx 0.57721$ is Euler's
constant.

As a simple consequence we derive the correlation between two
variables in the same group in the random effects model
(\ref{eq.oneway.cdf}). Suppose $X_i= \mu + \tau + G_i, \;i=1,2$ with
$\tau \sim$ExpS$(\alpha,0,\sigma), \; G_i \sim$Gumbel$(0,\sigma)$
and the three variables independent. Then $\mathrm{Cov} (X_1, X_2) =
\mathrm{Var}(\tau)$ and $\mathrm{Var}(X_i) =
\mathrm{Var}(\tau)+\mathrm{Var}(G_i)$. Since
$\mathrm{Var}(G_i)={\pi^2 \sigma^2 \over 6}$ we obtain that
$\mathrm{Cor}(X_1, X_2)=1-\alpha^2$, which varies from 0 in the
independent case $\alpha =1$ to 1 as $\alpha \to 0$, which is
reasonable since the limit corresponds to full dependence.


\section{Mixtures of generalized extreme value distributions}\label{sect:evmixtures}

The mixture models for the Gumbel distribution discussed so far in
the paper carry over to the (generalized) EV distribution in a
straightforward manner. However, the interpretation (i) is
different.

The EV distribution has d.f. $\exp(-(1+\gamma {{x-\mu} \over
\sigma})^{-1/\gamma})$ with parameters $\mu, \gamma \in R$ and
$\sigma>0$.  For positive $\gamma$ this distribution has a finite left
endpoint $\delta = \mu-\sigma/\gamma$ and for $\gamma$ negative it has a finite
right endpoint $\delta = \mu+\sigma/|\gamma|$.
In analogy with (\ref{stable}) - (\ref{xunconditional}) let
$S$ be positive stable with Laplace transform (\ref{stable}) and
assume that
\begin{equation}
\label{EVconditional}
P(X \leq x|S) = \exp[-S(1+\gamma {{x-\mu} \over \sigma})^{-1/\gamma}]
= \exp[-(\gamma {{x-\delta} \over S^\gamma \sigma})^{-1/\gamma}].
\end{equation}
Then by (\ref{stable}),
\begin{equation}
\label{EVunconditional}
P(X \leq x)= \exp\left[-\left\{1+(\gamma/\alpha) {{x-\mu} \over
(\sigma/\alpha)}\right\}^{-1/(\gamma/\alpha)}\right].
\end{equation}

Thus, in the terminology of (ii) of Section \ref{sect:mixtures}, if
$X$ is a positive stable size mixture of an EV distribution with
location $\mu$, scale $\sigma$ and shape parameter $\gamma$ then also
X itself has an EV distribution with the same location $\mu$ and the
same right endpoint $\delta$, but with a new scale parameter $\sigma/\alpha$
and new shape parameter $\gamma/\alpha$. Hence in particular the
unconditional distribution of $X$ has heavier tails than the
conditional one.

The physical motivations (ii) and (iii) from Section
\ref{sect:mixtures} carry over to the present situation without
change.  Further, from (\ref{EVconditional}) it can be seen that $X$
may be obtained as a special random location-scale transformation of
an EV distribution. Specifically, if $E$ has an EV distribution with
parameters $\mu, \sigma, \gamma$ and $S$ is positive $\alpha$-stable
and independent of $E$, then $X$ may be represented as
\begin{equation}
\label{eq:scalemixture}
X = S^\gamma E + (1-S^\gamma)\delta.
\end{equation}
 Thus $X$ is obtained as a scale mixture with mixing distribution
$S^\gamma$, but in addition there is an accompanying location change
which is tailored to keep the endpoint of the distribution
unchanged. This, of course, may be the most natural way to make scale
mixtures of distributions with finite endpoints.

With this change, the motivations from Section \ref{sect:mixtures} and
the models from Section \ref{sect:newclasses} carry over to the EV distribution.  If the
models in Section \ref{sect:newclasses} are written as size mixtures, i.e. in the form
(ii), the only changes needed to go from Gumbel to EV are to replace
$e^{- \frac{x-\mu}{\sigma}}$ by $(1+\gamma {{x-\mu} \over
\sigma})^{-1/\gamma}$ in all expressions.  The recursions for the
likelihood functions from Section 5 translate to the EV case
similarly.

It is also straightforward to translate specifications using (i) to
the EV case. E.g, in the formulation (i) the random effects model
(\ref{eq.oneway.cdf}) becomes
$$
X_{i, j} = S_i^\gamma E_{i,j}  + (1-S_i^\gamma)\delta,
$$ where $E_{i, j}$ has an EV distribution with parameters $\mu,
\sigma, \gamma$ and $S_i$  positive $\alpha$-stable, and all variables
are mutually independent. In the same way, the hidden time series model
(\ref{linear}) in EV form can be written as
$$
X_t = H_t^\gamma E_t  + (1-H_t^\gamma)\delta,
$$
with $H_t$ a linear stable process and $E_t$ is EV distributed, and all
variables are mutually independent.

Next,
$$
\log(X-\delta) = \gamma \log S + \log (E-\delta),
$$
and if $X$ is of the form (\ref{eq:scalemixture}) with $\gamma>0$ then
$\log(E-\delta)$ has a Gumbel distribution with location
parameter $\log (\sigma/\mu)$ and scale parameter $\gamma$. For
$\gamma <0$ we instead write
$$
\log(\delta-X) = \gamma \log S + \log (\delta-E),
$$
where $\log (\delta-E)$ has a Gumbel distribution with location
parameter $\log (\sigma/\mu)$ and scale parameter $\gamma$. Thus the
diagnostic plots for Gumbel mixtures could be used also for EV
mixtures, except that $\delta$ isn't known. A pragmatic way to
control the model assumptions then is to replace $\delta$ by some suitable
estimate.


\section{discussion}\label{sect:discussion}

The pitting corrosion example discussed in Section 4 was the
starting point for the present research. There it seemed important
to use models where marginal, conditional and unconditional
distributions, and maxima over blocks of varying sizes all had
Gumbel distributions, since this leads to simple and understandable
results, and credible extrapolation into extreme tails.

However it seems important to stay within the extreme value
framework throughout for many other applications too. This is a main
reason for the present work. Another is that our results open up a
wide spectrum of hitherto unavailable possibilities to
construct extreme value models for complex observation structures,
in particular for time series and spatial extreme value data.

The results also throw new light on some much studied logistic
models. In particular they point to possibilities for new kinds of
model diagnostics. In addition they show how one can carry over many
of the analyses available for normal models to an extreme value
framework in a simple and intuitive way. One example of how this can
be done is the suggested next step in the analysis of the corrosion
data, to fit a model which includes both random group means and a
MA(1) dependence.

We believe that many applications of these ideas remain to be explored.
One aim of this paper is to provide a solid basis
for such future research.

\vspace*{5mm} \noindent {\bf Acknowledgement:} We want to thank
Sture Holm for many stimulating discussions and ideas. We also want
to thank two anonymous referees for very useful remarks.
Acknowledgements to R-project and programs.  Research supported by
VCC/Ford, the Wallenberg foundation, and the Swedish Foundation for
Strategic Research.

\noindent Corresponding author: Anne-Laure \ Foug\`eres, { \'Equipe
Modal'X, B\^at. G, Universit\'e Paris X - Nanterre, 200 av. de la
R\'epublique, F-92000 Nanterre, France. Email:
Anne-Laure.Fougeres@u-paris10.fr}

\newpage

\begin{figure}
\begin{center}
\includegraphics[width=\textwidth]{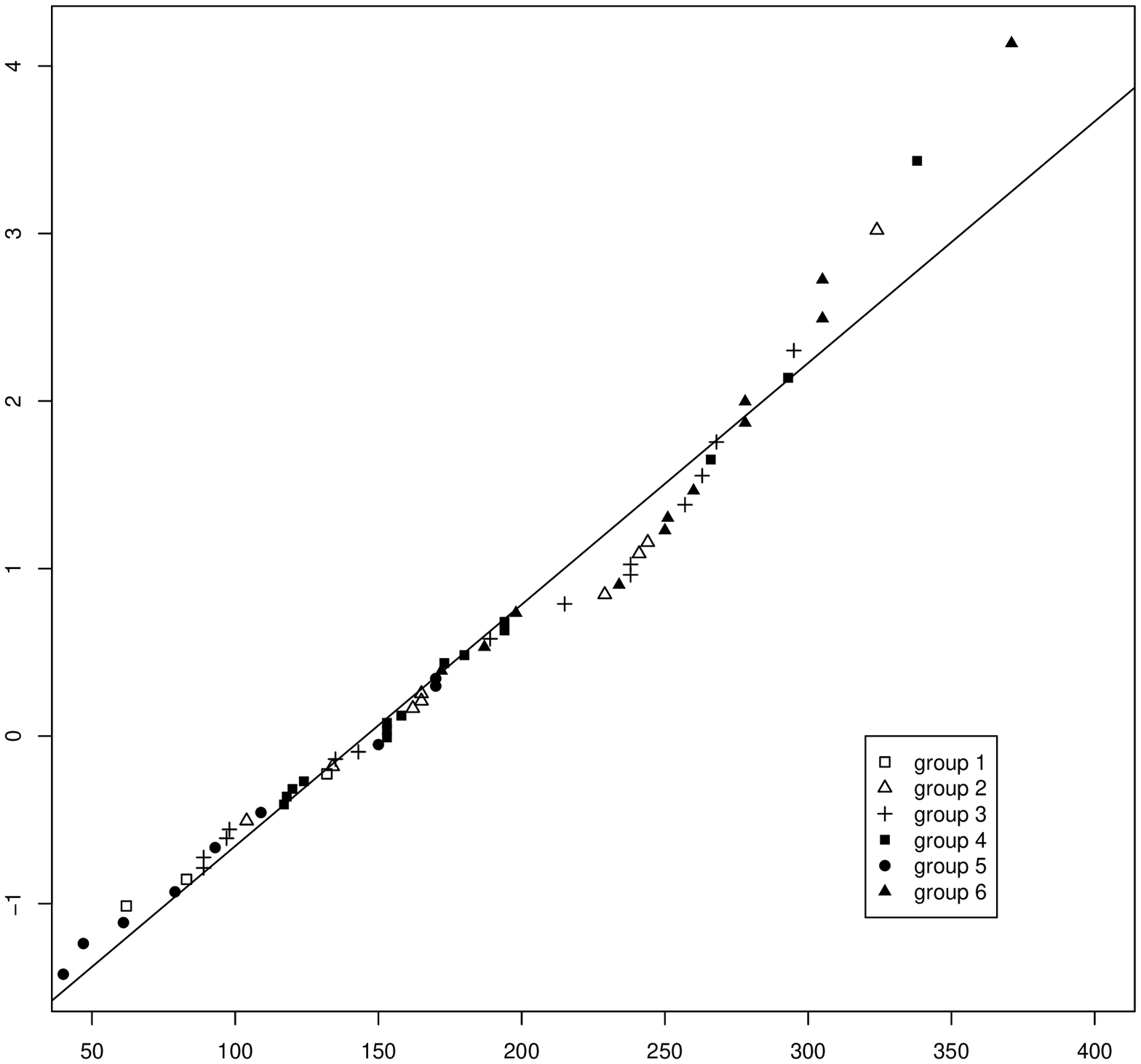}
\end{center}
\caption{Gumbel plot for the pooled corrosion measurements. A
different symbol is used for each group.} \label{fig:poolgumbel}
\end{figure}

\newpage

\begin{figure}
\begin{center}
\includegraphics[width=\textwidth]{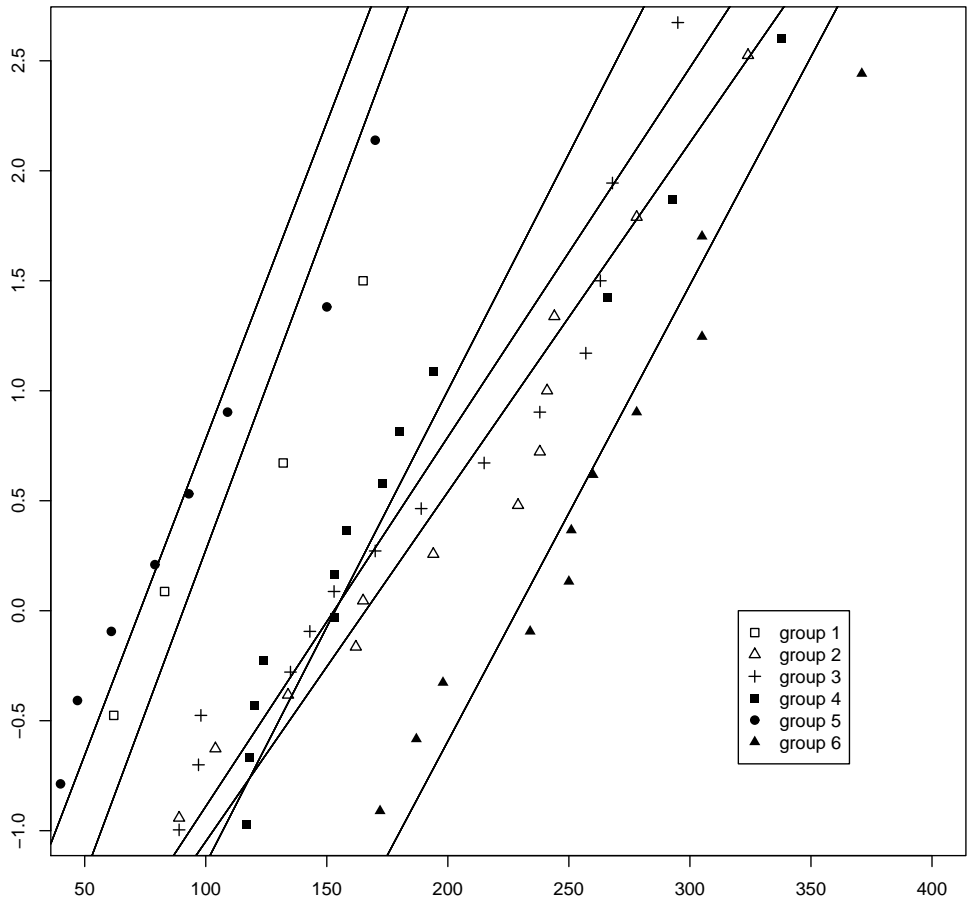}
\end{center}
\caption{Gumbel plots made separately for the 6 groups. The solid
lines are the different theoretical Gumbel fits for each group.}
\label{fig:groupgumbel}
\end{figure}

\newpage

\begin{figure}
\begin{center}
\includegraphics[width=\textwidth]{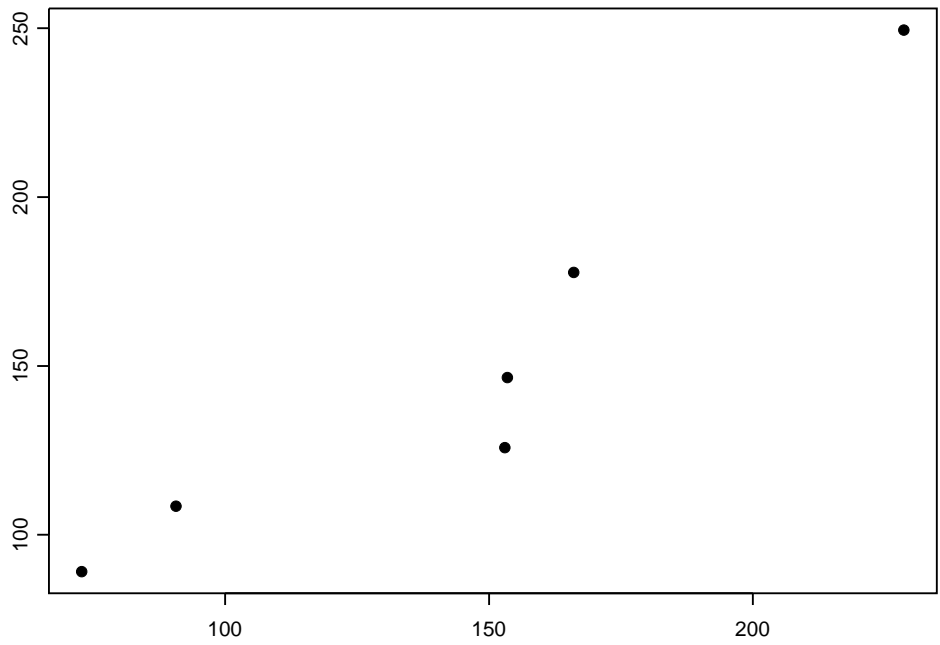}
\end{center}
\caption{qq-plot of fitted exponential-stable distribution against
estimated $\mu$-s from the conditional analysis with the same
$\sigma$ in all groups.} \label{fig:qqmu.stable}
\end{figure}

\newpage

\begin{figure}
\begin{center}
\includegraphics[width=\textwidth]{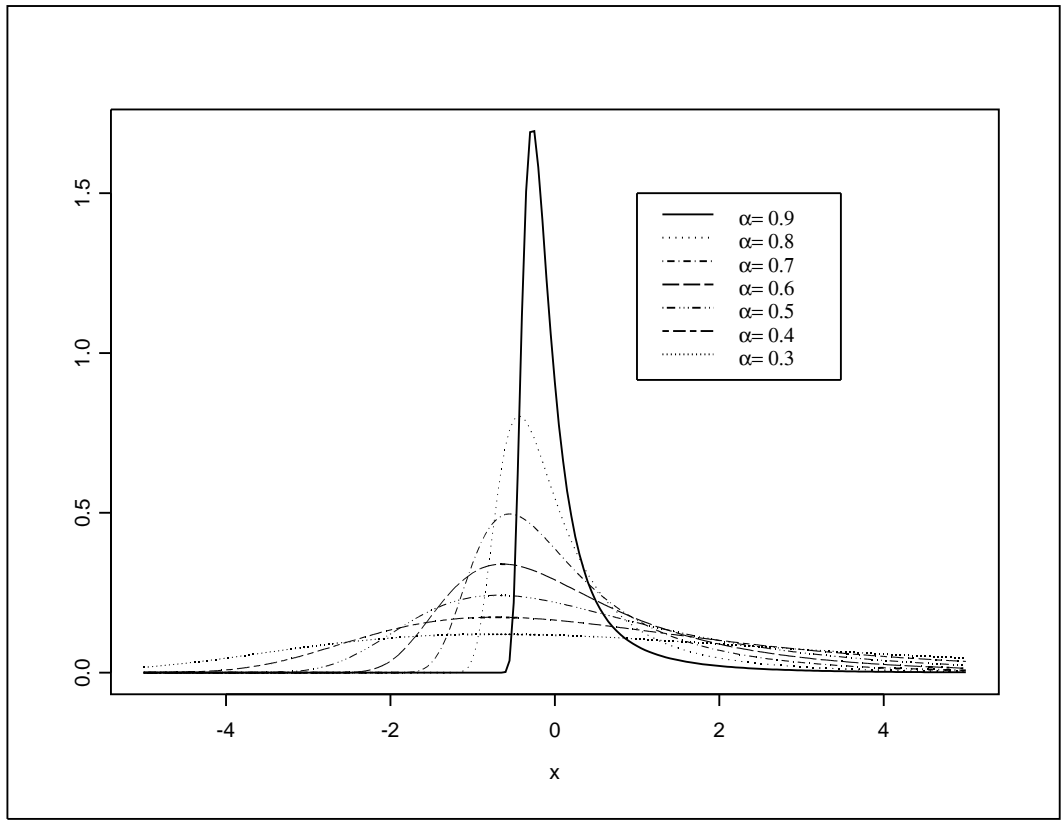}
\end{center}
\caption{Plot of densities of standardized exponential-stable
distributions ExpS$(\alpha,0,1)$, with varying $\alpha$.}
\label{fig:expstable.pdf}
\end{figure}

\end{document}